\documentclass[a4paper]{amsart}

\usepackage{amsmath}
\usepackage{amssymb}
\usepackage{amsfonts}
\usepackage{amsthm}

\newtheorem{lemma}{Lemma}[section]
\newtheorem{theorem}{Theorem}

\theoremstyle{definition}

\theoremstyle{definition}
\newtheorem{remark}[lemma]{Remark}
\theoremstyle{definition}

{\catcode `\@=11 \global\let\AddToReset=\@addtoreset}
\AddToReset{equation}{section}

\newcommand{\D}{\mathrm{d}}

\newcommand{\R}{\mathbb{R}}
\newcommand{\N}{\mathbb{N}}

\newcommand{\e}{{\varepsilon }}

\newcommand{\norm}[1]{ \left| \! \left| #1 \right| \! \right| }

\def\tr{\mathop{\rm tr}\nolimits} 

\begin{document}

\title[Classical limit for semi-relativistic Hartree systems]
{Classical limit for semi-relativistic Hartree systems}

\author[G. L. Aki] {Gonca L. Aki}

\author[P. A. Markowich]{Peter A. Markowich}

\author[C. Sparber]{Christof Sparber}

\address[G. L. Aki] {Faculty of Mathematics, University of Vienna, Nordbergstra\ss e 15, 
A-1090 Vienna, Austria}\email{gonca.aki@univie.ac.at}

\address[P. A. Markowich]
{Department of Applied Mathematics and Theoretical Physics,
University of Cambridge, Wilberforce Road, Cambridge CB3 0WA \&
Faculty of Mathematics, University of Vienna, Nordbergstra\ss e
15, A-1090 Vienna, Austria} \email{peter.markowich@univie.ac.at}

\address[C. Sparber]
{Department of Applied Mathematics and Theoretical Physics,
University of Cambridge, Wilberforce Road, Cambridge CB3 0WA}
\email{c.sparber@damtp.cam.ac.uk}

\subjclass[2000]{35Q40, 35Q55, 81S30, 81V17}
\keywords{semi-relativistic Hartree energy, classical limit, Wigner transform, relativistic Vlasov-Poisson system}

\begin{abstract}
We consider the three-dimensional semi-relativistic Hartree model for fast quantum mechanical particles moving in a self-consistent field. 
Under appropriate assumptions on the initial density matrix as a (fully) mixed quantum state we 
prove, using Wigner transformation techniques, that its classical limit yields the well known relativistic 
Vlasov-Poisson system. The result holds for the case of attractive and repulsive mean-field interaction, 
with an additional size constraint in the attractive case.
\end{abstract}

\thanks{This work has been supported by the KAUST Investigator Award and the 
Wolfson Research Merit Award (Royal Society) of P. Markowich. 
G. L. Aki acknowledges support by the DEASE project of the EU. 
C.~Sparber has been supported by the APART grant of 
the Austrian Academy of Sciences.}

\maketitle
\section{Introduction and main result}\label{sint}

In this paper we aim to establish the classical limit as $\e \to 0_+$ of the \emph{semi-relativistic Hartree system} (or Schr\"odinger-Poisson system), i.e.
\begin{equation}\label{equationrho}
\left \{
\begin{split}
i \varepsilon \partial_t \rho^\varepsilon(t) =& \, [H^\e, \rho^\varepsilon(t)],\quad x \in \R^3, \, t\in \R_+,\\
- \kappa \Delta  V^\e = & \, n^\varepsilon (t,x), \quad \kappa = \pm 1, 
\end{split}
\right.
\end{equation}
subject to an initial data $\rho^\e(0) \equiv \rho^\e_{0}$, 
where the reason for choosing $\kappa = \pm 1$ respectively will be explained below and 
the Hamiltonian operator $H^\e$ is given by 
\begin{equation*}
H^\e:=\sqrt{-\varepsilon^2\Delta+1}+  V^\varepsilon(t,x).
\end{equation*}
Here the pseudo-differential operator for the kinetic energy is simply defined via 
multiplication in Fourier space with the symbol $\sqrt{ |\e \xi|^2+1}$, 
for the semi-classical parameter $0<\e \ll 1$, a dimensionless scaled Planck's constant (all other physical constants are rescaled to be equal to $1$). 
This scalar pseudo-differential operator is frequently used in relativistic quantum mechanical models as a convenient replacement 
of the full (matrix-valued) Dirac operator.
In \eqref{equationrho} we denote by $\rho^\e(t)\in \mathfrak S_1(L^2(\R^3))$ the density matrix operator of the system, i.e. a positive self-adjoint
trace class operator acting on $L^2(\R^3)$. The particle density $n^\varepsilon(t,\cdot) \in L^1(\R^3)$ is then obtained by evaluating the 
corresponding kernel $\rho^\e(t,x,y)$, which, by abuse of notation is denoted by the same symbol as the density operator, 
on its diagonal \cite{Si}, i.e. $n^\varepsilon(t,x) = \rho^\e(t,x,x)$. 
The system \eqref{equationrho} describes the \emph{mean field dynamics} of (relativistic) quantum mechanical particles 
in a \emph{mixed state}. Denoting by $\{\psi_j^\varepsilon\}_{j\in \mathbb N}$ a complete orthonormal basis 
of $L^2(\R^3)$ we can decompose the kernel of $\rho^\e(t)\in \mathfrak S_1(L^2(\R^3))$ via 
\begin{equation*}
\rho^\varepsilon(t,x,y)=\sum_{j\in \N} \lambda_j^\varepsilon \, \psi_j^\varepsilon(t,x)\overline{{\psi}_j^\varepsilon}(t,y),
\end{equation*}
with (constant in time) $\lambda_j \in \ell^1$, $\lambda_j \geq 0$ \cite{Si}. 
Using this representation, we arrive at an \emph{equivalent} 
system of countable many nonlinear Schr\"odinger-type equations 
\begin{equation}\label{RSE}
\left \{
\begin{split}
i \varepsilon \partial_t\psi_j^\varepsilon= & \,
{\sqrt{-\varepsilon^2\Delta+1}}\ \psi_j^\varepsilon+V^\varepsilon(t,x)\psi_j^\varepsilon\\
-\kappa \Delta V^\e = & \, n^\varepsilon(t,x), 
\end{split}
\right.
\end{equation}
where the density is now given by
\begin{equation*}
n^\varepsilon(t,x)= \sum_{j \in \N} \lambda_j^\varepsilon|\psi_j^\varepsilon(t,x)|^2.
\end{equation*}
The system \eqref{RSE} can be interpreted as the time-dependent model associated 
to the semi-relativistic \emph{Hartree energy}
\begin{equation*}
\begin{split}
\mathcal E^\e_{\rm H}(t)= & \ \mathcal E^\e_{\rm kin}(t) + \mathcal E^\e_{\rm pot}(t)\\
= & \ \sum_{j\in \N} \lambda_j^\e {\|(-\e^2 \Delta +1)^{1/4} \psi_j^\e  \|}^2_{L^2(\R^3)}+ 
\frac{\kappa}{8 \pi }\iint_{\mathbb R^6} \frac{n^\varepsilon(t,x) n^\varepsilon(t,y)}{|x-y|} \, \D x \, \D y,
\end{split}
\end{equation*}
where for the second line we have used the three dimensional Green's function representation of the potential, i.e. 
\begin{equation*}
V^\varepsilon(t,x)=  \frac{\kappa}{4\pi |x|}*n^\varepsilon(t,x).
\end{equation*}
For future references we recall that the kinetic and the self-consistent potential-energy can also be 
shortly written in terms of density matrices as
\begin{equation*}
\mathcal E^\e_{\rm kin}(t) = \tr (\sqrt{-\varepsilon^2\Delta+1}\ \rho^\e(t) ), \quad 
\mathcal E^\e_{\rm pot}(t) = \frac{\kappa}{2}\tr (V^\e(t,x) \rho^\e(t) ).
\end{equation*}
In the case $\kappa = -1$, the coupling to the Poisson equation comprises an \emph{attractive nonlinearity} 
for the Schr\"odinger-type equations \eqref{RSE} and hence, global well-posedness for general initial data does not hold, cf. \cite{FrLen,Len}. 
In this case, the system \eqref{RSE} is a generalization (for mixed states) of the 
semi-relativistic Hartree model derived in \cite{ElSch} as the mean field limit for large systems of 
Bosons with gravitational self-interaction. This model has been extensively 
studied in recent years as it is considered to describe the dynamics of so-called \emph{Boson stars} 
\cite{FrJoLen, FrLen, FrLen1, Len}. In the case $\kappa = 1$, the Poisson interaction is \emph{repulsive} 
and typically models electron-electron self-interactions (in the Hartree-approximation). The system \eqref{RSE} therefore 
allows to study relativistic corrections to the usual Hartree model of many-body electron system as 
it is needed for example in the case of heavy atoms, cf. \cite{Lie}. 

In the present work we are interested in rigorously establishing the classical limit of \eqref{RSE} as $\e \to 0$. 
To this end we shall heavily rely on the Wigner transformed picture of quantum mechanics \cite{Wi}. 
To this end, we define the ($\e$-scaled) \emph{Wigner transform} of a given density matrix kernel $\rho^\varepsilon(t,x,y)$ as in 
\cite{MaMa, LiPa}
\begin{equation*}
f^\varepsilon[\rho^\e(t)] \equiv f^\varepsilon(t,x,\xi):=\frac{1}{(2\pi)^3}\int_{\mathbb{R}^3} \rho^\varepsilon
\left(t,x+\frac{\e y}{2} ,x-\frac{\e y}{2} \right)e^{-i\xi\cdot y} \, \D y,
\end{equation*}
From this definition we easily infer
\begin{equation}\label{scaling}
{\| f^\e(t,\cdot,\cdot) \|}^2_{L^2(\R^3_x\times \R^3_\xi)}= (2\pi\varepsilon)^{-3}{\| \rho^\e(t,\cdot,\cdot) \|}^2_{L^2(\R^3_x\times \R^3_y)}.
\end{equation}
Since $f^\varepsilon(t,x,\xi)$ is real-valued it can be seen as a quantum mechanical analog of 
a classical phase space distribution. However, the Wigner function in general also takes negative values and therefore does not allow 
for a probabilistic interpretation. 
Nevertheless, Wigner functions have proved to be a highly successful tool in the rigorous mathematical 
derivation of classical limits, see, e.g. \cite{Bo,GLM, GeMaMaPo, LiPa, MaMa, Pu, SpMaMa} for various analytical results and \cite{MaPiPo} for a numerical study. 
In particular, under appropriate uniform bounds on $\rho^\e$ (see Theorem 1 below), it is known that $f^\e$ has accumulation 
points $f\equiv f^0(t,x,\xi)$ as $\e \to 0$ (in a suitable topology), which are positive Borel measures on the phase space. 
The limiting measures can then be identified as a distributional solution of the corresponding limiting 
evolutionary system. In our case we expect it to be a solution of the 
\emph{relativistic Vlasov-Poisson system}, i.e.
\begin{equation}\label{VP}
\left \{
\begin{split}
& \, \partial_t f +\frac{\xi}{\sqrt{|\xi|^2+1}}\cdot \nabla_x f-\nabla_x V\cdot\nabla_\xi f= 0,\\
& \, -\kappa \Delta V =  n(t,x), 
\end{split}
\right.
\end{equation}
where the limiting particle density $n(t,x)$ is obtained as
\begin{equation*}
n(t,x)=\int_{\mathbb{R}^3}f(t,x,\xi) \, \D  \xi.
\end{equation*}
The system \eqref{VP} is in itself an intensively studied model, in particular in the gravitational case $\kappa = -1$, 
see e.g. \cite{GlSch, GlSch1, HaRe} and the 
references given therein. The present work thus provides a rigorous connection between 
the quantum mechanical model \eqref{equationrho}, or equivalently\eqref{RSE}, and its classical analog \eqref{VP}. The novelty of the 
work lies in the fact that we have to establish suitable estimates based on the 
relativistic kinetic energy (instead of the usual one $-\frac{\e^2}{2}\Delta$) in order to be able to pass to the limit in the 
nonlinear potential $V^\e(t,x)$. The main result of this paper is as follows.
\begin{theorem}\label{th1}
Let $\e_0 < 1$ and assume that the initial data $\rho^\e_0 \in \mathfrak S_1(L^2(\R^3))$ 
is a density matrix operator, such that 
\begin{align*}\label{conditionA}
\sup_{\e \in (0, \e_0]} \left( \tr \rho^\e_0 +  \frac{1}{\e^3} \tr {(\rho_0^\e)}^2 + \tr \sqrt{-\varepsilon^2\Delta+1}\, \rho^\e_0 \right) < \infty.\tag{A}
\end{align*}
In the gravitational case $\kappa = -1$ we additionally assume
\begin{align*}\label{conditionB}
\frac{1}{\e^3} \tr {(\rho_0^\e)}^2 < \frac{C_*}{\tr \rho^\e_0},\tag{B}
\end{align*}
where $C_*>0$ is a fixed $\e$-independent constant to be computed in the following.
Then there exists a unique mild solution $\rho^\e  \in C([0,\infty);\mathfrak S_1(L^2(\R^3))$ of the relativistic 
Hartree system \eqref{equationrho} and its Wigner transform 
$f^\varepsilon[\rho^\e] \in C([0,\infty);L^2(\R_x^3 \times \R^3_\xi))$ 
converges, up to extraction of sub-sequences, 
\begin{eqnarray*}
f^\varepsilon[\rho^\e]\stackrel{\e\rightarrow 0 }{\longrightarrow} f \quad \text{in $C([0,\tau];\mathcal S'(\R^3_x \times \R^3_\xi) - w*)$,}
\end{eqnarray*}
for any $\tau< \infty$, where $f\in C([0,\tau];\mathcal M^+ (\R_x^3\times \R^3_\xi))\cap 
L^\infty([0,\tau];L^1\cap L^2(\R_x^3\times \R^3_\xi))$ is a distributional 
solution of the relativistic Vlasov-Poisson system \eqref{VP}. Here $\mathcal M^+$ is the space of positive bounded 
Borel measures on phase-space, equipped with the weak-* topology. 
\end{theorem}
The above given theorem shows that distributional solutions to \eqref{VP} can indeed be interpreted as the classical 
limit (on any compact time-interval) of solutions to \eqref{RSE}, or equivalently \eqref{equationrho}. 
To get more insight on Assumption \eqref{conditionA} we again use the decomposition 
\begin{equation*}
\rho^\e_0 = \sum_{j\in \N} \lambda_j^\varepsilon \, \psi_j^\varepsilon(x)\overline{{\psi}_j^\varepsilon}(y),
\end{equation*}
to obtain
\begin{equation*}
\tr \rho_0^\e = \sum_{j\in \N} \lambda_j^\varepsilon, \quad 
\tr {(\rho_0^\e)}^2 = \sum_{j\in \N} {(\lambda_j^\varepsilon)}^2.
\end{equation*}
We remark that $\tr \rho_0 ^\e$ is the total charge or mass of the particle system. Thus 
Assumption \eqref{conditionA} implies a uniform (in $\e$) bound on the total 
charge/mass  
and simultaneously requires $\sum_{j\in \mathbb N}(\lambda^\varepsilon_{j})^2 \simeq \e^3$, as $\varepsilon \rightarrow 0$. 
(For a construction of a density matrix $\rho_{0}^\e$, which satisfies 
this requirements for $\psi_j\in H^1(\R^3)$ we refer to \cite{MaMa}.) 
This condition is to be expected from 
earlier papers \cite{LiPa, MaMa} and prevents us from establishing our result in the case of pure 
initial states, i.e. $\rho^\e_0 = \psi_j^\varepsilon(x)\overline{{\psi}_j^\varepsilon}(y)$, 
or even finite combinations of pure states.  
Roughly speaking, the requirement of a totally mixed initial state 
is needed to ensure that the limiting particle density $n(t,x)$ 
is not too singular and thus can 
be successfully convolved with the Poisson kernel $\propto 1/|x|$. This is not clear a-priori as 
the limiting $f(t,x,\xi)$, and thus also the limiting density $n(t,x)$, in general is only a measure. 
In the case of a fully mixed state though we gain a bit more regularity, which is needed in passing to the limit $\e \to 0$. 
Note that for any Hilbert-Schmidt $\rho^\e \in \mathfrak S_2(L^2(\R^3))$, it holds 
$$\norm{\rho^\e}^2_{\mathfrak S_2(L^2(\R^3))}=\tr (\rho^\e)^2  = {\| \rho^\e(\cdot,\cdot) \|}^2_{L^2(\R^3_x\times \R^3_y)}< \infty \, ,$$ cf. \cite{Si}. 
Having in mind the scaling property \eqref{scaling}, Assumption \eqref{conditionA} therefore implies 
for the initial Wigner function 
${\| f^\e(0,\cdot,\cdot) \|}_{L^2} < \infty,$ uniformly in $\e$. This property is 
shown to be conserved by the time-evolution below and thus, yields an important uniform bound on $f^\e(t,x,\xi)$. 
Assumption \eqref{conditionB}, needed in the gravitational case, 
then additionally requires that the $L^2$-norm of the initial Wigner function or the total charge/mass 
is sufficiently small (and not only bounded), 
cf. Remark \ref{remexplain} below.\\

The rest of the paper is now organized as follows: In Section \eqref{s:prelim} we collect some preliminary results needed for the 
proof of our main theorem which is given in Section \ref{s:proof}. We then close this section with some final remarks on possible generalizations.

\section{Preliminary results}\label{s:prelim}

Let us start with a technical lemma, that nevertheless turns out to be crucial for the proof of the main Theorem. 

\begin{lemma} \label{lemma1}
Let $\rho(x, y)$ be the kernel of a positive self-adjoint trace class operator $\rho \in \mathfrak S_1(L^2(\R^3))$ 
and $n(x)\equiv \rho(x, x)$ the corresponding density. Then for $p\in[1,\infty)$ the following estimate holds
\begin{equation}\label{L-T}
{\|n  \|}_{L^{q}(\mathbb R^3)}\leq C_p\, \norm{\rho}_{\mathfrak S_p(L^2(\R^3))}^{\theta}
(\tr |\nabla|\rho )^{1-\theta},
\end{equation}
where
$$
q=\frac{4p-3}{3p-2},\qquad \theta=\frac{p}{4p-3}.
$$
\end{lemma} 
\begin{proof} For the proof we proceed analogously to \cite{LiPa}. Since the case $p=1$ is immediate we consider only $p>1$.
We start by considering the operator $|\nabla|-\beta n^\alpha$ where $\beta>0$ and $\alpha>0$ are some constants to be chosen later on. 
Let $\mu_1\leq\mu_2\leq\dots$ be the negative eigenvalues of the operator $|\nabla|-\beta n^\alpha$ and let $\{\varphi_k\}_{k \in \N}$ be a finite or countable 
collection of the corresponding ortho-normed eigenvectors. 
We consequently obtain
$$
\tr\ (|\nabla|-\beta n^\alpha )\rho\geq\sum_{k\in \N}\left\langle (|\nabla|-\beta n^\alpha )\rho \, \varphi_k,\varphi_k\right\rangle=
 \sum_{j,k\in \N}|p_{jk}|^2 \mu_k \lambda_j,
$$
where $p_{jk}=\int_{\mathbb R^3}\psi_j\overline{\varphi}_k$. 
Since $\mu_k = - |\mu_k|, \forall \, k \in \N$, we can rearrange this inequality in the following form
\begin{equation*}\label{ineq1}
\tr \beta n^\alpha \rho= \beta\int_{\mathbb R^3}n^{\alpha+1}\D x  \leq  \  \tr|\nabla|\rho + \sum_{j,k\in \N}|p_{jk}|^2 |\mu_k|\lambda_j.
\end{equation*}
Now, using H\"older's inequality, and the fact that $\sum_{j}|p_{jk}|^2=1$, as well as $\sum_{k}|p_{jk}|^2\leq 1$, we obtain
\begin{eqnarray}
\beta\int_{\mathbb R^3}n^{\alpha+1}\D x &\leq& \tr|\nabla|\rho + 
\norm{\rho}_{\mathfrak S_p(L^2(\R^3))}\left(\sum_{j\in \N}|\mu_j|^{p'}\right)^{1/{p'}}\nonumber\\
&\leq & \tr|\nabla|\rho + \norm{\rho}_{\mathfrak S_p(L^2(\R^3))}C_p\left(\int_{\mathbb R^3}(\beta n^\alpha)^{3+{p'}}\D x\right)^{1/{p'}}.\nonumber
\end{eqnarray}
Here the second inequality is obtained from Theorem 2.1 in \cite{FrLiSe}, which states that for all $\delta>0$
$$
\sum_{j\geq1}|\mu_j|^\delta\leq C_{\delta,d}\int_{\mathbb R^d}W(x)_+^{\delta+d}\D x,
$$
where $\mu_j$ are the negative eigenvalues of the operator $|\nabla|-W$. Now with the 
choice $\alpha=(p-1)/(3p-2)$, setting $\alpha+1=q$ and 
$$
\beta=\left(\frac{\left(\int_{\mathbb R^3}n^{q}\D x \right)^{1/p}}{{2C\norm{\rho}}_{\mathfrak S_p(L^2(\R^3))}}\right)^{{p'}/3},
$$
it is easy to conclude
$$
\left(\int_{\mathbb R^3}n^{q}\D x\right)^{1/q}\leq C_p \norm{\rho^\e}_{\mathfrak S_p(L^2(\R^3))}^{\theta}\, (\tr |\nabla|\rho)^{1-\theta}.
$$
\end{proof}
\begin{remark}  The proof can easily be generalized to the $d$-dimensional case,
where one finds
$$q=\frac{d(p-1)+p}{d(p-1)+1}\, ,\qquad \theta=\frac{p}{d(p-1)+p}\, .$$
\end{remark}

Next, we state a local-in-time existence result for the Hartree system \eqref{RSE}.

\begin{lemma} \label{existence}
Let $\rho^\e_{0}$ be such that Assumption \eqref{conditionA} holds. 
Then there exists a $T>0$ and unique mild solution $\rho^\e   \in C([0,T);\mathfrak S_1(L^2(\R^3))$ to \eqref{RSE}, 
or equivalently \eqref{equationrho}, which satisfies the following conservation laws:
\begin{equation}\label{cons}
\tr \rho^\e(t) = \tr \rho^\e_{0}, \quad \tr (\rho^\e(t))^2 = \tr (\rho^\e_{0})^2.
\end{equation}
\end{lemma}
\begin{proof} The existence of a unique solution can be obtained in a straightforward way 
from the results given in \cite{FrLen} (where the case of density matrices with finite rank is treated). 
Also, the conservation laws \eqref{cons} are established there. 
\end{proof}
Note that the conservation laws \eqref{cons} together with Assumption \eqref{conditionA} 
directly imply that, for all $\e\in (0, \e_0]$, it holds 
$${\|n^\e(t)\|}_{L^1(\R^3)}={\|n^\e(0)\|}_{L^1(\R^3)}<\infty,$$ 
and 
$${\|f^\e(t)\|}_{L^2(\R_x^3\times \R^3_\xi)}= {\|f^\e(0)\|}_{L^2(\R_x^3\times \R^3_\xi)} < \infty, $$
for all $t\in[0,T)$. 

\begin{lemma}\label{globalexistence} The solution stated above exists globally in-time, i.e. 
$T= \infty$, and we additionally have conservation of energy: 
$\mathcal E^\e_{\rm H}(t)=\mathcal E^\e_{\rm H}(0)$, provided that one of the following conditions is satisfied:
\begin{itemize}
\item[\emph{i)}] $\kappa=1$,
\item[\emph{ii)}]$\kappa=-1$ and Assumption \eqref{conditionB} is satisfied.
\end{itemize}
In particular we have the following uniform bound
\begin{equation}\label{bound}
\mathcal E^\e_{\rm kin}(t) \leq C,\quad \forall \, t\in [0,\infty).
\end{equation}
\end{lemma}
\begin{proof}
Formally, the conservation of energy can be obtained as in \cite{FrLen}.
In order to obtain the uniform a-priori bound \eqref{bound} we shall use 
\eqref{L-T}, with $p = 2$, which yields 
\begin{align*}
{\|n^\e (t) \|}_{L^{5/4}}\leq & \, C_2 {\|\rho^\e (t,\cdot, \cdot)\|}_{L^{2}}^{2/5}
(\tr |\nabla|\rho^\e(t) )^{3/5} \\
\leq & \, C_2\, \e^{-3/5}{\|\rho^\e (0,\cdot, \cdot)\|}_{L^{2}}^{2/5}
\left(\tr \sqrt{-\varepsilon^2\Delta+1}\ \rho^\e(t) \right)^{3/5},
\end{align*}
due to the second conservation law in \eqref{cons} and the fact that 
$\e|\nabla| \leq \sqrt{-\varepsilon^2\Delta+1}$, as can be seen on the level of the Fourier-symbols. 
On the other hand, recalling the Sobolev inequality 
$$
\iint_{\mathbb R^3 \times \R^3} \frac{n^\varepsilon(t,x) n^\varepsilon(t,y)}{|x-y|} \, \D x \, \D y \leq C_{\rm s}{\|n^\e (t) \|}^2_{L^{6/5}},
$$
and interpolating the density via 
${\|n^\e (t) \|}_{L^{6/5} }\leq \|n^\varepsilon(t)\|_{L^1}^{1/6} \, \|n^\varepsilon(t)\|_{L^{5/4}}^{5/6} $, 
we obtain 
$$
\iint_{\mathbb R^3 \times \R^3} \frac{n^\varepsilon(t,x) n^\varepsilon(t,y)}{|x-y|} \, \D x \, \D y \leq \tilde{C}{\mathcal E_{\rm kin}^\e (t)}, 
\quad \forall \, t\in[0,\infty),
$$ 
where $\tilde{C}=C_{\rm s} C_2^{5/3} {(\tr \rho_0^\e)}^{1/3}  {\| f^\e(0,\cdot,\cdot)\|}_{L^2}^{2/3}$ is $\e$-independent by 
Assumption \eqref{conditionA}. 
On the one hand, this shows that the initial energy is indeed well defined and uniformly bounded in $\e$ for both $\kappa = \pm 1$. Moreover, we immediately 
conclude the uniform boundedness of the kinetic energy in the case $\kappa=1$ and hence, global-in-time existence of solutions. In the case 
$\kappa=-1$ we only get, from energy conservation, that 
$$
{\mathcal E_{\rm H}^\e (0)}={\mathcal E_{\rm kin}^\e (t)}-\frac{1}{8\pi}
\iint_{\mathbb R^3 \times \R^3} \frac{n^\varepsilon(t,x) n^\varepsilon(t,y)}{|x-y|} \, \D x \, \D y \geq \left(1-\frac{\tilde{C}}{8\pi}\right){\mathcal E_{\rm kin}^\e (t)}.
$$
Uniform boundedness of ${\mathcal E_{\rm kin}^\e (t)}$, and hence, global-in-time existence, can therefore be concluded 
if $\tilde{C}< 8\pi$, which holds true under Assumption \ref{conditionB}.
\end{proof}

\begin{remark}\label{remexplain} 
In comparison to the proof given in \cite{Len} we needed to invoke slightly different arguments in order to obtain the uniform (w.r.t. $\e$) 
bound \eqref{bound}. The requirement $\tilde{C}< 8\pi$ directly leads to Assumption \eqref{conditionB} with $C_*= \frac{(8\pi)^{3}}{C_{\rm s}^{3}C_2^{5}}$. 
Note however, that in our mixed-state case we can no longer characterize the condition $\tilde{C}< 8\pi$ by the solution of the 
corresponding single-state ground state problem as it is done in \cite{Len}. 
\end{remark}


\section{Proof of the Main Result}\label{s:proof}

In this paper we shall use the following definition for the Fourier transform
$$
(\mathcal{F}_{\xi\rightarrow \eta}\varphi)(\eta) \equiv \hat \varphi (\eta)= \int_{\mathbb R^3}\varphi (\xi)e^{-i \xi \cdot \eta}\, \D \xi. 
$$
Moreover we denote by $\mathcal{S}$ the Schwartz space of 
rapidly decaying functions.

\begin{proof} The proof of the Theorem \ref{th1} consists of three steps:\\

{\it Step 1.} We first note that due to Assumption \eqref{conditionA} and 
the conservation laws \eqref{cons}, the Wigner function $f^\e(t)$ is uniformly bounded 
in $L^2(\R^3_x\times \R^3_\xi)$ and in $\mathcal S'(\R^3_x\times \R^3_\xi)$, where the later 
follows from the results in \cite{GeMaMaPo}. 
Thus, for every fixed $t\in [0,\infty)$, it holds (up to extraction of sub-sequences): $f^\e(t)\rightharpoonup f(t)$, as $\e \to 0$, 
in $\mathcal{S}(\mathbb{R}_x^3\times \mathbb{R}_\xi^3) - w*$ and in $L^2(\mathbb{R}_x^3\times \mathbb{R}_\xi^3)-w$. 
Moreover it has been shown in \cite{GeMaMaPo, LiPa} that the limit $f\in \mathcal{M}^+(\mathbb{R}_x^3\times \mathbb{R}_\xi^3)$, i.e. 
a positive phase-space measure.\\

{\it Step 2.} Next we shall prove the time-equicontinuity of $f^\e(t,x,\xi)$. 
To this end, we have to show that $\partial_t f^\varepsilon$ is bounded in $L^\infty((0,\tau),\mathcal S'(\R^3_x\times \R^3_\xi))$, 
for any $\tau < \infty$. 
We consequently consider the (phase space) weak formulation of the Wigner transformed evolution equation, i.e. 
\begin{equation}\label{weak}
-\left\langle\partial_t f^\varepsilon,\phi \right\rangle=  \langle\Gamma^\varepsilon f^\varepsilon,\phi\rangle + \langle \Theta^\varepsilon[V^\varepsilon]f^\varepsilon,\phi\rangle,
\end{equation}
where $\phi\in \mathcal S(\mathbb{R}_x^3\times \mathbb{R}_\xi^3)$ and we write  
$\left\langle~\cdot~,~\cdot~ \right\rangle$ for the corresponding duality bracket on $\mathcal S$. 
Above, $\Gamma^\e$ and $\Theta^\e[V^\varepsilon]$ are pseudo-differential operators, corresponding to the Wigner transformed 
kinetic and potential energy terms. Explicitly, $\Gamma^\varepsilon f^\varepsilon$ is given by
$$
\Gamma^\varepsilon f^\varepsilon(t,x,\xi)=\frac{i}{(2\pi)^3} \iint_{\mathbb{R}^3 \times \mathbb{R}^3}e^{iy\cdot(x-\eta)}\ y\cdot \gamma^\varepsilon(y,\xi)f^\varepsilon(t,\eta,\xi)\, {\rm d}y \, {\rm d}\eta,
$$
where
$$
\gamma^\varepsilon(y,\xi):=\frac{2\xi}{\sqrt{|\xi+\varepsilon\frac{y}{2}|^2+1}+\sqrt{|\xi-\varepsilon\frac{y}{2}|^2+1}} \in \R^3.
$$
On the other hand, the nonlinear potential $V^\e(t,x)$ enters via the well known operator, cf. \cite{LiPa, MaMa}
$$
\Theta^\varepsilon[V^\varepsilon] f^\varepsilon(t,x,\xi) 
=\frac{i}{(2\pi)^3}\iint_{\mathbb{R}^3\times \R^3} e^{i\eta\cdot(z-\xi)}\delta^\e(t,x,\eta) f^\varepsilon(t,x,z)\, {\rm d}\eta \, {\rm d}z.\nonumber
$$
whose symbol is given by
\begin{align*}
\delta^\e(t,x,\eta) := \frac{1}{\e} \left(V^\varepsilon(t,x+\varepsilon\frac{\eta}{2}) - V^\varepsilon(t,x-\varepsilon\frac{\eta}{2})\right) = \eta\cdot\int_{-1/2}^{1/2}\nabla_xV^\varepsilon(t,x+\varepsilon s\eta)\, \D s.
\end{align*}
In order to estimate the first term on the r.h.s. of \eqref{weak}, we use an inverse Fourier transform w.r.t. $x$ 
and Plancherel's theorem to obtain
\begin{align*}
|\langle \Gamma^\varepsilon f^\varepsilon,\phi\rangle| 
 & = \Big|\iint_{\mathbb{R}^6}\Big( \int_{\mathbb{R}^3}e^{-iy\cdot\eta}(\mathcal F^{-1}_{x\rightarrow y}{\nabla_x\phi})(y,\xi) \cdot 
 \gamma^\varepsilon(\xi,y)\, {\rm d}y \Big) f^\varepsilon(t,\eta,\xi){\rm d}\eta \, {\rm d}\xi\Big|\\
& \leq {\|f^\varepsilon(t)\|}_{L^2}\Big(\iint_{\mathbb{R}^6}\left| (\mathcal F^{-1}_{x\rightarrow y}{\nabla_x\phi})(y,\xi) \right|^2 \left|\gamma^\varepsilon(\xi,y)\right|^2 {\rm d}y \, {\rm d}\xi\Big)^{1/2}\\
&
\leq{\|f^\varepsilon(t)\|}_{L^2}{\|\nabla_x \phi \|}_{L^2},
\end{align*}
where the second inequality follows from the fact that 
$|\gamma^\e|\leq1$, for all $\e \in (0,\e_0]$. Since $f^\e(t,x,\xi)$ is uniformly 
bounded in $L^2(\R^3_x\times \R^3_\xi)$ for all $t\in [0,\infty)$, due to \eqref{cons} and Assumption \eqref{conditionA}, 
we obtain the desired uniform bound for $|\langle \Gamma^\varepsilon f^\varepsilon,\phi\rangle|$.
Next, we consider the nonlinear term on the r.h.s. of \eqref{weak}. 
After a Fourier transform w.r.t. $\xi$, we need to estimate
\begin{align*} 
|\langle\Theta^\varepsilon[V^\varepsilon] f^\varepsilon,\phi\rangle| =  
\Big|\iint_{\mathbb{R}^6}\hat{\phi}(x,\eta)\Big(\int_{-1/2}^{1/2}\eta\cdot \nabla_x V^\varepsilon(t,x+\varepsilon s\eta){\rm d}s\Big)
(\mathcal F^{-1}_{x\rightarrow y}{f}^\varepsilon)(t,x,\eta){\rm d}x \, {\rm d}\eta\Big |.
\end{align*}
Using the generalized Young inequality we have
$$
{\|\nabla_xV^\varepsilon(t)\|}_{L^2(\mathbb R^3)}\leq C {\Big \|\nabla_x {|x|^{-1}} \Big \|}_{L^{3/2}_w(\mathbb R^3)} 
{\|n^\varepsilon(t)\|}_{L^{6/5}(\mathbb R^3)},
$$
and by interpolation between $L^1(\R^3)$ and $L^{5/4}(\R^3)$, one obtains 
\begin{equation*}
\|n^\varepsilon(t)\|_{L^{6/5}}\leq \|n^\varepsilon(t)\|_{L^{1}}^{1/6}\, \|n^\varepsilon(t)\|_{L^{5/4}}^{5/6} \leq 
\tr(\rho^\varepsilon(t))^{1/6} \, \|f(t,\cdot, \cdot)^\varepsilon\|_{L^{2}}^{1/3} 
(\mathcal E_{\rm kin}^\e(t))^{1/2}.
\end{equation*}
Due to the conservation laws \eqref{cons} and the uniform bound on the kinetic energy \eqref{bound} we therefore find 
that $\|\nabla_xV^\varepsilon(t)\|_{L^2}$ is uniformly bounded as $\varepsilon\rightarrow 0$, for all $t\in [0,\infty)$. 
Thus, we can estimate
\begin{align*}
|\langle \Theta^\varepsilon[V^\varepsilon] f^\varepsilon,\phi\rangle| & \leq   
{\|f^\varepsilon(t)\|}_{L^2}\Big(\iint_{\mathbb{R}^6}|\hat{\phi}(x,\eta)|^2|\eta|^2\Big|\int_{-1/2}^{1/2}\nabla_xV^\varepsilon(t,x+\varepsilon s\eta){\rm d}s\Big|^2 {\rm d}x \, {\rm d}\eta\Big)^{1/2}\\
&\leq 
{\|f^\varepsilon(t)\|}_{L^2} {\|\nabla_x V^\varepsilon(t)\|}_{L^2}
\Big(\int_{\mathbb{R}^3}\sup_{y\in \R^3}|\hat{\phi}(y-\varepsilon s\eta,\eta)|^2|\eta|^2 {\rm d}\eta\Big )^{1/2}\\
&\leq
{\|f^\varepsilon(t)\|}_{L^2}\, {\|\nabla_xV^\varepsilon(t)\|}_{L^2}\, \| \, {\sup_{y\in \R^3} |\hat{\phi}(y,\eta)||\eta|}\,\|_{L^2} < \infty,
\end{align*}
uniformly in $\e$. Thus, $\partial_t f^\e$ is bounded in $L^\infty((0,\tau), \mathcal S'(\mathbb{R}_x^3\times \mathbb{R}_\xi^3))$, uniformly in $\e$, and 
we conclude that (again up to extraction of sub-sequences)
\begin{eqnarray}
f^\varepsilon \stackrel{\e\rightarrow 0 }{\longrightarrow} f \quad &\textnormal{in}& \quad 
C((0,\tau);\mathcal S'(\mathbb{R}_x^3\times \mathbb{R}_\xi^3) - w*),\nonumber\\
\textnormal{and}  &\textnormal{in}& \quad L^\infty((0,\tau);L^2(\mathbb{R}_x^3\times \mathbb{R}_\xi^3))-w*.\nonumber
\end{eqnarray}
Together with the results of {\it Step 1} this implies 
$$
f \in C((0,\tau);\mathcal M^+ (\mathbb{R}_x^3\times \mathbb{R}_\xi^3) - w*)\cap L^\infty((0,\tau);L^1\cap L^2(\mathbb{R}_x^3\times \mathbb{R}_\xi^3) ).$$ 

{\it Step 3.} It remains to identify the limiting $f(t,x,\xi)$, by passing to the limit in the (full) 
weak formulation of the Wigner transformed evolution equation, i.e.
\begin{equation}\label{weakevol}
\begin{split}
& \int_0^\infty   \iint_{\mathbb R^6} \big(-\partial_t \sigma (t) \phi(x,\xi) - \sigma(t)\Gamma^\e  \phi(x,\xi) \big) 
f^\e(t,x,\xi) \, \D x \, \D  \xi \, \D t\\
&  + \int_0^\infty   \iint_{\mathbb R^6} \sigma(t) \Theta^\e[V^\e]  \phi (x,\xi)  f^\e(t,x,\xi) \, \D x \, \D  \xi \, \D t = 
 \iint_{\mathbb R^6} \sigma(0) \phi(x,\xi) f_0^\e(x,\xi) \, \D x \, \D \xi , 
\end{split}
\end{equation}
where $\sigma\in C_0^\infty(\mathbb R^+_t)$ and $\phi\in\mathcal S(\mathbb{R}_x^3\times \mathbb{R}_\xi^3)$. First it is easily seen that 
$$
\lim_{\varepsilon \rightarrow 0}\sup_{t\in [0,\tau]}|\langle f^\varepsilon,
\Gamma^\varepsilon \phi\rangle-\langle f,\nabla_x\phi\cdot \frac{\xi}{\sqrt{|\xi|^2+1}}\rangle|=0,
$$
since 
$f^\varepsilon \rightharpoonup f$ in $L^\infty((0,\tau];L^2(\mathbb{R}_x^3\times \mathbb{R}_\xi^3))-w*$, as shown in {\it Step 2}, and the convergence of the rest of the duality bracket is strong in ${C([0,T];L^2(\mathbb{R}_x^3\times \mathbb{R}_\xi^3))}$. Concerning the nonlinear term, we have to show 
\begin{equation}\label{nlincon}
\lim_{\varepsilon \rightarrow 0}\sup_{t\in [0,\tau]}|\langle f^\varepsilon, \Theta^\varepsilon[V^\varepsilon] \phi\rangle-\langle  f, \nabla_xV \cdot \nabla_\xi\phi\rangle|=0.
\end{equation}
Similarly to what is done above, after invoking a Fourier transform w.r.t. $\xi$, it is sufficient to prove that
\begin{align*}
\lim_{\varepsilon\rightarrow 0}\Big(\hat{\phi}(x,\eta) \int_{-1/2}^{1/2}\eta\cdot \nabla_x V^\varepsilon(t,x+\varepsilon s\eta){\rm d}s\Big) = 
 \hat{\phi}(x,\eta) \eta\cdot\nabla_xV(t,x),
\end{align*}
in ${C([0,T];L^2(\mathbb{R}_x^3\times \mathbb{R}_\xi^3))}$ strongly. 
In order to do so we 
denote the compact support of $\hat{\phi}$ by $\Omega:\subseteq \Omega_x\times \Omega_\eta$ and write
\begin{align*}
&\int_{-1/2}^{1/2}  \Big(
\iint_\Omega|\hat{\phi}(x,\eta) |\eta||^2|\nabla_xV^\e(t,x+\varepsilon s\eta)
 - \nabla_xV(t,x)|^2 {\rm d}x \, {\rm d}\eta \Big )^{1/2}
{\rm d}s \\
&\leq \sup_{x,\eta\in\Omega}|\hat{\phi}(x,\eta) |\eta||\int_{-1/2}^{1/2} \Big (\iint_\Omega|\nabla_xV^\varepsilon(t,x+\varepsilon s\eta)
 - \nabla_xV(t,x+\varepsilon s\eta)|^2 {\rm d}x \, {\rm d}\eta \Big. \\
 & \qquad\qquad\qquad\qquad\qquad \Big.+ 
\iint_\Omega|\nabla_xV(t,x+\varepsilon s\eta)
 - \nabla_xV(t,x)|^2 {\rm d}x \, {\rm d}\eta ~\Big)^{1/2}
{\rm d}s \\
&=: \sup_{x,\eta\in\Omega}|\hat{\phi}(x,\eta) |\eta||\int_{-1/2}^{1/2}( I^\varepsilon  + J^\varepsilon )^{1/2} \, {\rm d}s.
\end{align*}
First, we consider the term $J^\varepsilon $: Using Plancherel's Theorem we can write
\begin{eqnarray}
J^\varepsilon &=&
\iint_\Omega |\nabla_x V(t,x+\varepsilon s\eta)
 - \nabla_x V(t,x)|^2 {\rm d}x \, {\rm d}\eta
\nonumber\\
&=&
\iint_\Omega \left|(\mathcal F_{x\rightarrow z}\nabla_x V)(t,z)\right|^2| \, e^{i\varepsilon s\eta\cdot z} -1 |^2{\rm d}x \, {\rm d} \eta
,\nonumber
\end{eqnarray}
from which we conclude that $J^\varepsilon\rightarrow 0$, as $\e \to 0$, by dominant convergence. 
In order to treat the term $I^\e$ we take into account that $V^\varepsilon$ solves the Poisson equation
\begin{equation*}\label{poisson}
- \kappa \Delta  V^\e =  \, n^\varepsilon(t,x), \quad \kappa = \pm 1
\end{equation*}
with $n^\varepsilon(t)\in L^{5/4}(\mathbb R^3)\cap L^1(\mathbb R^3)$, uniformly in $\e$. Due to the 
regularizing property of the Poisson equation we have $\nabla_xV^\varepsilon\in W_{\rm loc}^{1,5/4}({\mathbb R}^3)$ uniformly in $\e$, 
which is compactly embedded in $L^2_{\rm loc}({\mathbb R}^3)$. We therefore infer that, for all $t \in [0,\tau]$, it holds
\begin{equation}\label{gradientV}
\nabla_xV^\varepsilon(t)
 \stackrel{\e\rightarrow 0 }{\longrightarrow} \nabla_xV(t) \quad\textnormal{in}\quad L^2(\mathcal O),
\end{equation}
where $\mathcal O=\{z:=x+\e s\eta\in\mathbb R^3|\ (x,\eta)\in(\Omega_x^3\times\Omega_\eta^3)\}$. We finally 
recall the result of Theorem III.2 in \cite{LiPa} 
which (up to extraction of subsequences) states that, as $\e \to 0$,  
$$n^\varepsilon(t,x)\rightharpoonup n(t,x)=\int_{\mathbb{R}^3}f(t,x,\xi){\rm d}\xi, $$ 
in $C([0,T];\mathcal M^+(\mathbb{R}_x^3)-w*)$.
We therefore conclude $I^\varepsilon\rightarrow 0$, as $\e \to 0$, which consequently implies \eqref{nlincon}. 
In summary we have shown that the Wigner transformed evolution equation \eqref{weakevol}
converges weakly to 
\begin{equation*}
\begin{split}
 & \int_0^\infty   \iint_{\mathbb R^6} \Big(-\partial_t \sigma (t) \phi(x,\xi) - \frac{\xi}{\sqrt{|\xi|^2+1}} \cdot \nabla_x 
 \phi(x,\xi) \sigma(t) \Big) f(t,x,\xi) \, \D x \, \D  \xi \, \D t \\
& \, + \int_0^\infty \iint_{\mathbb R^6} \sigma(t) \nabla_x V \cdot \nabla_\xi  \phi (x,\xi)   f(t,x,\xi) \, \D x \, \D  \xi \, \D t 
 = \iint_{\mathbb R^6} \sigma(0) \phi(x,\xi) f_0(x,\xi) \, \D x \, \D \xi  , 
\end{split}
\end{equation*}
i.e. the relativistic Vlasov-Poisson system (in distributional sense). This finishes the proof. \end{proof}

\begin{remark} It is straightforward to generalize our result for \emph{Yukawa-type interactions}, i.e. 
$$
V^\varepsilon(t,x)= \kappa \, \frac{ e^{-\lambda |x|}}{4\pi |x|}*n^\varepsilon(t,x),\quad \lambda > 0,
$$
which are used in, e.g. \cite{Len}. Indeed it is possible to prove our theorem for an even more general class of 
self-consistent interaction potentials $V^\varepsilon=G* n^\varepsilon$, if the kernel $G=G(x)$ satisfies appropriate 
regularity conditions, cf. \cite{Bo, LiPa} for a discussion in the case of non-relativistic Hartree systems. 
Finally, we remark that one can also treat the case of \emph{semi-relativistic Hartree-Fock systems}, cf. \cite{FrLen}, 
by combining straightforwardly the results present here with those given in \cite{GIMS} for non-relativistic Hartree-Fock systems. 
To this end, one should note that the so-called \emph{exchange-term} vanishes as $\e \to 0$. Hence, 
one again recovers the relativistic Vlasov-Poisson system \eqref{VP} in the limit.
\end{remark}


\bibliographystyle{amsplain}

\end{document}